\documentclass[a4paper]{amsart}
\usepackage{amssymb,amsmath,amsthm}
\usepackage{hyperref}
\usepackage[foot]{amsaddr}
\usepackage{cleveref}

\makeatletter
\def\paragraph{\@startsection{paragraph}{4}%
  \z@\z@{-\fontdimen2\font}%
  {\normalfont\bfseries}}
\makeatother

\title[A homotopy-theoretic model of functional extensionality in $\Eff$]{A homotopy-theoretic model of function extensionality in the effective topos}
\date{\today}
\author{Daniil Frumin$^1$}
\address{${}^1$ ICIS, Radboud University Nijmegen, the Netherlands. E-mail: dfrumin@cs.ru.nl}
\author{Benno van den Berg$^2$}
\address{${}^2$ Institute for Logic, Language and Computation, Universiteit van Amsterdam, P.O. Box 94242, NL, 1090 GE Amsterdam, the Netherlands. E-mail: bennovdberg@gmail.com.}
\usepackage{ast2}
\usepackage{bussproofs}
\calclayout

\newcommand{\leibniz}{\hat{\otimes}}

\newcommand{\sym}{\mathsf{s}}
\newcommand{\tr}{\mathsf{tr}}
\newcommand{\st}{\mathsf{st}}
\newcommand{\rel}{\mathsf{rel}}
\newcommand{\sv}{\mathsf{sv}}
\newcommand{\tot}{\mathsf{tot}}

\newcommand{\N}{\mathbf{N}}

\begin{document}

\maketitle

\begin{abstract}
  We present a way of constructing a Quillen model structure on a full subcategory of an elementary topos, starting with an interval object with connections and a certain dominance.
  The advantage of this method is that it does not require the underlying topos to be cocomplete.
  The resulting model category structure gives rise to a model of homotopy type theory with identity types, $\Sigma$- and $\Pi$-types, and functional extensionality.

  We apply the method to the effective topos with the interval object $\nabla 2$.
  In the resulting model structure we identify uniform inhabited objects as contractible objects, and show that discrete objects are fibrant.
  Moreover, we show that the unit of the discrete reflection is a homotopy equivalence and the homotopy category of fibrant assemblies is equivalent to the category of modest sets.
  We compare our work with the path object category construction on the effective topos by Jaap van Oosten.
\end{abstract}

\section{Introduction}

Any constructive proof implicitly contains an algorithm; realizability makes this algorithm explicit. For instance, realizability shows how from a constructive proof of a statement of the form
\[ \forall x \in \mathbb{N} \, \exists y \in \mathbb{N} \, \varphi(x, y) \]
one can extract an algorithm computing a suitable $y$ given $x$ as input. For this reason, realizability, as invented by Kleene in 1945 \cite{KleeneSC:intint}, has become an important tool in the study of formal systems for constructive mathematics. More recently, it has been used to provide semantics for various type theories, in particular polymorphic type theories for which no set-theoretic models exist. For these more advanced applications a category-theoretic understanding of realizability is essential; indeed, around 1980 Martin Hyland discovered the effective topos, a topos whose internal logic is given Kleene's realizability. In fact, various realizability interpretations exist besides Kleene's original variant, and many of these interpretations can be given a topos-theoretic formulation (for more on this, see \cite{vanOosten:realiz}). Here it has to be understood that these realizability toposes are elementary toposes in the sense of Lawvere and Tierney: they are not
Grothendieck toposes. In particular, realizability toposes are not cocomplete (for instance,
in the effective topos the countable coproduct of $1$ does not exist), a point which will be important for us later.

The purpose of this paper is to make some first steps in applying ideas from realizability to homotopy type theory. Homotopy type theory refers to a recent influx of ideas from abstract homotopy theory and higher category category to type theory. The starting point for these developments is the discovery by Hofmann and Streicher \cite{hofmannstreicher98} that Martin-L\"of's identity type gives every type in type theory the structure of a groupoid; in fact, they give every type the structure of an $\infty$-groupoid as shown in \cite{typesWeakOmGpd,lumsdaine10}. Conversely, types in type theory can be interpreted as $\infty$-groupoids: this is what underlies Voevodsky's interpretation of type theory in simplicial sets, with the types  interpreted as Kan complexes \cite{kapulkinetal12}. Such a Kan complex is both understood as a combinatorial model for the homotopy type of a space (hence "homotopy type theory") and a notion of $\infty$-groupoid. In his proof Voevodsky relies heavily on the fact that the category of simplicial sets carries a Quillen model structure in which the Kan complexes are precisely the fibrant objects. This model structure is also essential for the interpretation of the identity types, as in \cite{HtpyModelsId}.

Homotopy theory not only provides an unexpected interpretation of type theory, but it also gives one a new perspective on some old problems in type theory, such as the mysterious identity types and the problem of extensional constructs \cite{hofmann1995extensional}. In addition, it suggests many new extensions of type theory, such as higher-inductive types and the univalence axiom. In this paper we will focus on a particular consequence of univalence: function extensionality. For more on these exciting new developments we refer to the \cite{hottbook}.

So far realizability has only played a minor role in these developments (some exceptions are \cite{htpyEff,cubicalrealizability}). However, given its prominent place in the study of constructive formal systems it seems quite likely that realizability will be fruitful here as well. Also, the most important questions in homotopy type theory (such as Voevodsky's Main Computational Conjecture) concern its computational behaviour. Since realizability
aims to make the computational content of constructive formal systems explicit, a realizability interpretation of homotopy type theory would help us understand its computational content.

In this paper we make a step in that direction by endowing a subcategory of the effective topos with a Quillen model structure. In fact, in the first half of this paper we show that in any elementary topos equipped with a suitable class of monomorphisms and an interval object one can define three classes of maps  (cofibrations, fibrations and weak equivalences, respectively) such that on the full subcategory of fibrant objects these induce a model structure. This theorem subsumes the classical model structure on simplicial sets in which the fibrant objects are the Kan complexes and it also includes the work on cubical sets by Coquand and others \cite{CCHM}.

This result is inspired by earlier work by Orton and Pitts \cite{pittsAxioms} and Gambino and Sattler \cite{UniformFibrations}, which in turn is based on the work of Cohen, Coquand, Huber, and M\"ortberg \cite{CCHM} and earlier work by Cisinski \cite{Cisinski200243}; indeed, with a few exceptions most steps in our proof of the model structure can be found in these earlier sources. The main innovation  is that we do not assume cocompleteness of the underlying topos, so that our result can be applied to realizability toposes as well. This means that, unlike Cisinski, Gambino and Sattler, we do not rely on the small object argument to build our factorisations.

As we already mentioned, any model structure gives rise to a model of Martin-L\"of's identity types. But they also provide a model for strong $\Sigma$-types and products. Again following Gambino and Sattler, we also show that in our setting one can interpret $\Pi$-types, which moreover satisfy function extensionality. Function extensionality says that  two functions are equal if they give equal outputs on identical inputs, and this is one of those desirable principles  which are valid on the homotopy-theoretic interpretation of type theory, but are unprovable in type theory proper. So in our setting we are able to interpret basic type theory with function extensionality. (Here we, like many authors,
ignore coherence issues related to substitution: for a possible solution, see \cite{lumsdainewarren15}.)

In some more detail, the precise contents of the first few sections are as follows. In Section 2 we recall some important categorical notions (like that of a model structure, a dominance and the Leibniz adjunction) that will be used throughout this paper. In Section 3 we present our axiomatic set-up for building model structures. We define cofibrations, fibrations and (strong) homotopy equivalences in this setting and we establish some basic properties of these classes of maps. This is then used in Section 4 to construct a model structure on the full subcategory of fibrant objects. We also show that the resulting model of type theory interprets extensional $\Pi$-types.

In the second part of this paper we apply this general recipe for constructing model structures to Hyland's effective topos. For our class of monomorphisms we take the class of all monos and for our interval object we take $\nabla 2$. The latter choice is inspired by earlier work by Jaap van Oosten \cite{htpyEff}. It also seems natural, because $\nabla 2$ contains two points, which are, however, computationally indistinguishable
(because they have identical realizers). The analogy with the usual interval $[0,1]$ is that its endpoints are distinct, but homotopy-theoretically indistinguishable.

So in Section 5 we recall some basic facts about the effective topos and check that it fits into our axiomatic framework. In Section 6 we make some progress in characterizing contractible objects and maps in the effective topos and we show that these are closely related to the uniform objects and maps. This also leads to a concrete criterion for characterizing the fibrant assemblies. In Section 7 we prove that discrete objects like the natural numbers object are fibrant and we show that the homotopy category of the full subcategory of fibrant assemblies is the category of modest sets.

In the same section we also compare our work with earlier work by Jaap van Oosten \cite{htpyEff}. In his paper Van Oosten constructs a path object category structure (in the sense of \cite{pathObjCat}) on the effective topos, resulting in a type-theoretic fibration category in the sense of Shulman \cite{shulman13a}. This falls short of a full model structure, but it does provide an interesting interpretation of the identity types. The main difference is that in Van Oosten's structure every object is fibrant and function extensionality does not hold (private communication). So our paper constructs
the first homotopy-theoretic model of function extensionality in the effective topos.

The present paper is written purely in the language of category theory and does not assume familiarity with homotopy type theory. We use {\bf ZFC} as our metatheory and are aware that some of our results in the section on the effective topos make use of the axiom of choice. We leave it for future work to determine what can said in a constructive metatheory (but see \cite{MscThesis}).

The contents of this paper are based on the Master thesis of the first author written under supervision of the second author (see \cite{MscThesis}). We thank Jaap van Oosten for useful comments on the thesis.


\section{Categorical definitions}

In this section we recall, for the convenience of the reader, the definitions of a model structure, a dominance and the Leibniz adjunction.

\begin{defi}{liftingproperties} Let $f$ and $g$ be two morphisms in some category \ct{C}.
We will say that $f$ has the \emph{left lifting property (LLP)} with respect to $g$ and $g$ has the \emph{right lifting property (RLP)} with respect to $f$, and write $f \pitchfork g$, if for any commuting square in \ct{C}
\diag{ D \ar[d]_f \ar[r] & B \ar[d]^g \\
C \ar[r] & A }
there exists a map $h: C \to B$ (a \emph{diagonal filler}) making the two resulting triangles commute. If $\mathcal A$ is some class of morphisms in \ct{C}, we will write $\mathcal A^\pitchfork$ for the class of morphisms in \ct{C} having the RLP with respect to every morphism in $\mathcal A$, and ${}^\pitchfork \mathcal A$ for the class of morphisms in \ct{C} having the LLP with respect to every morphism in $\mathcal A$.
\end{defi}

\begin{defi}{WFS} A \emph{weak factorisation system (WFS)} on a category \ct{C} is a pair $({\mathcal L}, {\mathcal R})$ consisting of two classes of maps in \ct{C} such that
\begin{enumerate}
\item every map $h$ in \ct{C} can be factored as $h = gf$ with $f \in \mathcal{L}$ and $g \in \mathcal{R}$.
\item $\mathcal{L}^\pitchfork = \mathcal{R}$ and ${}^\pitchfork \mathcal{R} = \mathcal{L}$.
\end{enumerate}
\end{defi}

\begin{lemm}{retractargument} {\rm (Retract argument)} A pair $(\mathcal{L}, \mathcal{R})$ consisting of two classes of maps in a category \ct{C} is a weak factorisation system if and only if the following condtions hold:
\begin{enumerate}
\item every map $h$ in \ct{C} can be factored as $h = gf$ with $f \in \mathcal{L}$ and $g \in \mathcal{R}$.
\item for any $l \in \mathcal{L}$ and $r \in \mathcal{R}$ one has $l \pitchfork r$.
\item both $\mathcal{L}$ and $\mathcal{R}$ are closed under retracts.
\end{enumerate}
\end{lemm}
\begin{proof}
See, for instance, Lemma 11.2.3 in \cite{Riehl:cathtpy}.
\end{proof}

\begin{defi}{modelstructure} A \emph{(Quillen) model structure} on a category \ct{C} consists of three classes of maps $\mathcal{C}, \mathcal{F}$ and $\mathcal{W}$, referred to as the cofibrations, the fibrations and the weak equivalences, respectively, such that the following hold:
\begin{enumerate}
\item both $(\mathcal{C} \cap \mathcal{W}, \mathcal{F})$ and $(\mathcal{C}, \mathcal{F} \cap \mathcal{W})$ are weak factorisation systems.
\item in any commuting triangle
\diag{ ^C \ar[rr]^f \ar[dr]_h & & B \ar[dl]^g \\
& A }
if two of $f,g,h$ are weak equivalences, then so is the third. (This is called 2-out-of-3 for weak equivalences.)
\end{enumerate}
\end{defi}

\begin{defi}{dominance} Let \ct{E} be an elementary topos and $\Sigma$ be a class of monomorphisms in \ct{C}. Then $\Sigma$ is called a \emph{dominance} if
\begin{enumerate}
\item every isomorphism belongs to $\Sigma$ and $\Sigma$ is closed under composition.
\item every pullback of a map in $\Sigma$ again belongs to $\Sigma$.
\item the category $\Sigma_{cart}$ of morphisms in $\Sigma$ and pullback squares between them has a terminal object.
\end{enumerate}
\end{defi}

One can show, using standard arguments, that for the terminal object $m: A \to B$ in $\Sigma_{cart}$ we must have $A = 1$ and $B \subseteq \Omega$. We will also write $\Sigma$ for the codomain of this classifying map, so that the classifying map is written $\top: 1 \to \Sigma$. This map $\top: 1 \to \Sigma$ is a pullback of the map $\top: 1 \to \Omega$ classifying all monomorphisms and determines the entire class. Indeed, a dominance can equivalently be defined as a subobject $\Sigma \subseteq \Omega$ satisfying the following principles in the internal logic of \ct{E}:
\begin{enumerate}
\item $\top \in \Sigma$.
\item $(\forall p, q \in \Omega) \, ( \, (p \in \Sigma \land (p \Rightarrow (q \in \Sigma))) \Rightarrow p \land q \in \Sigma \, )$.
\end{enumerate}

\begin{prop}{WFSfromdominance}
If $\Sigma$ is a dominance on an elementary topos \ct{E}, then $(\Sigma, \Sigma^\pitchfork)$ is a weak factorisation system.
\end{prop}
\begin{proof} This seems to be well-known (e.g. \cite[Section 4.4]{AWFS1}), but since we have not been able to locate this precise theorem in the literature, we include some details here.
Both $\Sigma$ and $\Sigma^\pitchfork$ are closed under retracts, so by the retract argument (\reflemm{retractargument}) it suffices to prove that any map $h: B \to A$ can be factored as a map in $\Sigma$ followed by a map in $\Sigma^\pitchfork$. This can be done as follows:
\diag{ B \ar[rr]_(.35)f & & \Sigma_{a \in A} \Sigma_{\sigma \in \Sigma} B_a^\sigma \ar[rr]_(.65)g & & A, }
with $f(b) = (h(b), \top, \lambda x.b)$ and $g(a, \sigma, \tau) = a$.
Here $B^{\sigma}$ denotes an object of maps $\{\ast \mid \sigma \} \to B$.
Note that in the case of $A = 1$, the object $\Sigma_{\sigma \in \Sigma} B^{\sigma}$ is isomorphic to the object $\widehat{B}$ of \cite[Section 3.1]{rosolini} representing $\Sigma$-partial maps.
By \cite[Proposition 3.1.3]{rosolini}, the inclusion $B \hookrightarrow \Sigma_{\sigma \in \Sigma} B^\sigma$ is a $\Sigma$-map; and the map $\Sigma_{\sigma \in \Sigma} B^\sigma \to 1$ has the right lifting property against $\Sigma$-maps by \cite[Proposition 3.2.4]{rosolini}. For a general case, a similar argument is performed in the slice over $A$.

Details are left to the reader.
\end{proof}

\begin{defi}{Leibnizadjunction} Suppose $f: A \to B$ and $g: C \to D$ are two maps in an elementary topos \ct{E}. Then the \emph{Leibniz product} (or \emph{pushout product}) of $f$ and $g$ is the unique map $f \hat\otimes g$ making
\diag{ A \times C \ar[r]^{f \times 1} \ar[d]_{1 \times g} & B \times C \ar@/^1pc/[rdd]^{g \times 1} \ar[d] \\
A \times D \ar[r] \ar@/_1pc/[rrd]_{f \times 1} & \bullet \ar@{.>}[dr]|-{f \hat\otimes g} \\
& & B \times D}
commute with the square being a pushout.

The \emph{Leibniz exponential} of $f:A \to B$ and $g: C \to D$ is the unique map $\hat\exp(f,g)$ making
\diag{ C^B \ar@/^1pc/[rrd]^{C^f} \ar@{.>}[dr]|-{\hat\exp(f,g)} \ar@/_1pc/[ddr]_{g^B} \\
& D^B \times_{D^A} C^A \ar[r] \ar[d] & C^A \ar[d]^{g^A} \\
& D^B \ar[r]_{D^f} & D^A }
commute with the square being a pullback.
\end{defi}

\begin{prop}{Leibnizadjunction} {\rm (Leibniz adjunction)} The operations $\hat\otimes$ and $\hat\exp$ define bifunctors on $\ct{E}^{\to}$, and give rise to a two-variable adjunction
\[ \ct{E}^\to(f \hat\otimes g,h) \cong \ct{E}^\to(f,\hat\exp(g,h)). \] 
Also, for any triple of maps $f, g, h$ we have $h \pitchfork \hat\exp(f, g)$ if and only if $f \hat\otimes h \pitchfork g$.
\end{prop}
\begin{proof}
See Exercise 11.1.9 and Lemma 11.1.10 in \cite{Riehl:cathtpy}.
\end{proof}

\section{An axiomatic set-up}
\label{sec:setup}

In this section we will introduce our axiomatic set-up for building a model structure, which is inspired by earlier work by Orton and Pitts \cite{pittsAxioms} and Gambino and Sattler \cite{UniformFibrations}. Following Gambino and Sattler, we define four classes of maps (cofibrations, fibrations and (strong) homotopy equivalences) and establish their basic properties. The results and proofs contain few surprises for the homotopy theorist, but we include them here, because we need to make sure that they can be established without using cocompleteness of the underlying category.

The setting in which we will be working will be the following:
\begin{enumerate}
\item We are given an elementary topos \ct{E}.
\item Within this topos \ct{E} we are given an interval object $\mathbb{I}$, which here will mean that it comes equipped a monomorphism $[\partial_0, \partial_1]: 1 + 1 \to \mathbb{I}$ and \emph{connections} $\land, \lor: \mathbb{I} \times \mathbb{I} \to \mathbb{I}$ satisfying: \[ i \land 0 = 0 \land i = 0, \quad i \land 1 = 1 \land i = i, \quad i \lor 0 = 0 \lor i = i, \quad i \lor 1 = 1 \lor i = 1. \]
\item A class $\mathcal C$ of monomorphisms in \ct{E} satisfying the following axioms:
\begin{enumerate}
\item $\mathcal{C}$ is a dominance.
\item Elements in $\mathcal C$ are closed under finite unions (in other words, $\bot \in \Sigma$ and $p, q \in \Sigma \Rightarrow p \lor q \in \Sigma$).
\item The map $[\partial_0, \partial_1]: 1 + 1 \to \mathbb{I}$ belongs to $\mathcal C$.
\end{enumerate}
The elements of $\mathcal C$ will be referred to as the \emph{cofibrations}.
\end{enumerate}
It follows from these axioms that both $\partial_i$ are cofibrations and that the cofibrations are closed under Leibniz products.

\begin{example}
\begin{enumerate}
\item We could take for \ct{E} the category of simplicial sets, in which we have an internal given by $\Delta[1]$ and the class of all monomorphisms is a class of cofibrations. In this case the model structure we will construct is the classical Quillen model structure on the Kan complexes (see \cite{UniformFibrations}).
\item It would also be possible to take the category of cubical sets as in \cite{CCHM,bezem2014model}. As discussed in Orton and Pitts \cite{pittsAxioms}, this work fits into the present setting by taking for $\mathbb{I}$ the obvious representable and for $\mathcal C$ a special class of monos generated by the face lattice of \cite{CCHM}.
\end{enumerate}
\end{example}

Within this setting we make the following definitions.
\begin{defi}{trivfibrations} A morphism in ${\mathcal C}^\pitchfork$ will be referred to as a \emph{trivial fibration}.
\end{defi}

By \refprop{WFSfromdominance} we know that the cofibrations and trivial fibrations form a weak factorisation system on \ct{E}.

\begin{defi}{fibrations} A morphism $f$ in \ct{E} is a \emph{fibration} if it has the right lifting property with respect to maps of form $\partial_i \hat\otimes u$ with $u \in \mathcal C$ and $i \in \{ 0,1 \}$ (note that $\partial_i \hat\otimes u \in \mathcal C$, so that trivial fibrations are indeed fibrations). An object $X$ will be called \emph{fibrant} if the unique map $X \to 1$ is a fibration. We will write $\ct{E}_f$ for the full subcategory of \ct{E} consisting of the fibrant objects.
\end{defi}

Note that we are assuming that every map $0 \to X$ is a cofibration, so that every object in \ct{E} is cofibrant in that sense.

\begin{prop}{onusingleibnizadjunction}
\begin{enumerate}
\item If $u$ is a cofibration and $f$ is a (trivial) fibration, then $\hat\exp(u, f)$ is a (trivial) fibration as well.
\item A morphism $f$ is a fibration if and only if both $\hat\exp(\partial_i, f)$ are trivial fibrations.
\end{enumerate}
\end{prop}
\begin{proof} This is immediate from the Leibniz adjunction, the fact that the Leibniz product is associative and commutative, and the fact that cofibrations are closed under Leibniz products.
\end{proof}

\subsection{The homotopy relation}

\begin{defi}{homotopyrelation}
Two parallel arrows $f, g: B \to A$ will be called \emph{homotopic} if there is a morphism $H: \mathbb{I} \times B \to A$, a \emph{homotopy}, such that $f = H(\partial_0 \times B)$ and $g = H(\partial_1 \times B)$; in this case we will write $f \simeq g$, or $H: f \simeq g$, if we wish to stress the homotopy.
\end{defi}

\begin{prop}{homcongronfibrobj}
The homotopy relation defines a congruence on $\ct{E}_f$.
\end{prop}
\begin{proof} Since the homotopy relation is clearly preserved by pre- and postcomposition, it suffices to prove that the homotopy relation defines an equivalence relation on each homset between fibrant objects. In fact, to prove that the homotopy relation defines an equivalence relation on ${\rm Hom}_\ct{E}(Y, X)$ we will only need to assume that $X$ is fibrant ($Y$ can be arbitrary).

Let us first prove that the homotopy relation defines an equivalence on \[ {\rm Hom}_\ct{E}(1, X) \] for every fibrant object $X$. For any $p: 1 \to X$ we have the map
\diag{ r_p: \mathbb{I} \ar[r] & 1 \ar[r]^p & X, }
showing that the homotopy relation is reflexive.

To show that it is symmetric and transitive, we use that $X \to 1$ has the right lifting property with respect to the map
\[ [\partial_0,\partial_1] \hat\otimes \partial_0 = \mathbb{I} \times \{ 0 \} \cup \{ 0 \} \times \mathbb{I} \cup \{ 1 \} \times \mathbb{I} \subseteq \mathbb{I} \times \mathbb{I}. \]
So if $\alpha: \mathbb{I} \to X$ is a path in $X$, then there must be a map $H$ making
\diag{ \bullet \ar[d]_{[\partial_0,\partial_1] \hat\otimes \partial_0} \ar[rr]^{[r_{\alpha\partial_0},\alpha,r_{\alpha\partial_0}]} & & X \\
\mathbb{I} \times \mathbb{I} \ar[urr]_H }
commute. Then for $\beta = H(\mathbb{I} \times \partial_1)$, we have $\beta\partial_0 = \alpha\partial_1$ and $\beta\partial_1 = \alpha\partial_0$, showing that the homotopy relation is symmetric.

Similarly, if $\alpha, \beta: \mathbb{I} \to X$ are two paths with $\alpha\partial_1 = \beta\partial_0$, then we must have a map $K$ making
\diag{ \bullet \ar[d]_{[\partial_0,\partial_1] \hat\otimes \partial_0} \ar[rr]^{[\alpha,r_{\alpha\partial_0},\beta]} & & X \\
\mathbb{I} \times \mathbb{I} \ar[urr]_K }
commute. Then for $\gamma = K(\mathbb{I} \times \partial_1)$ we have $\gamma\partial_0 = \alpha\partial_0$ and $\gamma\partial_1 = \beta\partial_1$, showing that the homotopy relation is transitive.

Since for every fibrant object $X$ the map $\hat\exp(0 \to Y, X \to 1) = X^Y \to 1$ is a fibration by \refprop{onusingleibnizadjunction}.1, this shows that the homotopy relation is an equivalence relation on ${\rm Hom}_\ct{E}(1, X^Y)$ and hence on ${\rm Hom}_\ct{E}(Y, X)$ as well.
\end{proof}

\begin{defi}{homequivalence}
A morphism $f: B \to A$ is a \emph{homotopy equivalence} if there is a morphism $g: A \to B$, a \emph{homotopy inverse}, such that $gf \simeq 1_B$ and $fg \simeq 1_A$.
\end{defi}

\begin{coro}{homeqsatisfies2out3}
On $\ct{E}_f$ the homotopy equivalences satisfy 2-out-of-3 (indeed, they satisfy 2-out-of-6).
\end{coro}

\begin{prop}{htpyeq_closed_under_retracts}
Homotopy equivalences are preserved under retracts.
\end{prop}
\begin{proof}
    Let $g$ be a homotopy equivalence with a homotopy inverse $u$, and
  consider the following retract diagram:
\begin{displaymath}
\xymatrix{
    A \ar[r]^h \ar[d]_f & C \ar[d]^g \ar[r]^k & A \ar[d]^f \\
    B \ar[r]_l & D \ar@{.>}@/^/[u]^u \ar[r]_m & B
    }
\end{displaymath}
Then $(k \circ u \circ l)$ is a homotopy inverse of $f$, as witnessed by:
\[
k \circ u \circ l \circ f = k \circ u \circ g \circ h \simeq k \circ h = 1_A
\]
\[
f \circ k \circ u \circ l = m \circ g \circ u \circ l \simeq m \circ l = 1_B
\]
\end{proof}

\subsection{Strong homotopy equivalence}

\begin{defi}{strongheq}
A homotopy equivalence $f: B \to A$ together with homotopy inverse $g$ and homotopies $H: gf \simeq 1_B$ and $K: fg \simeq 1_A$ is called \emph{strong} if
\diag{ \mathbb{I} \times B \ar[r]^H \ar[d]_{\mathbb{I} \times f} & B \ar[d]^f \\
\mathbb{I} \times A \ar[r]_K & A }
commutes.
\end{defi}

In what follows it will be convenient to use an alternative characterisation of the strong homotopy equivalences. For this we should observe that for any $f: B \to A$ there are maps $\theta_f: f \to \partial_0 \hat\otimes f$ and $\sigma_f: \hat\exp(\partial_0,f) \to f$ in $\ct{E}^\to$:
\begin{displaymath}
\begin{array}{cc}
\xymatrix{ B \ar[d]_f \ar[rr]^{i_1(\partial_1 \times B)} & & A \cup_B \mathbb{I} \times B \ar[d]_{\partial_0 \hat\otimes f}  \\
A \ar[rr]_{\partial_1 \times A} & & \mathbb{I} \times A  }
&
\xymatrix{  B^\mathbb{I} \ar[d]_{\hat\exp(\partial_0,f)} \ar[rr]^{B^{\partial_1}} & & B \ar[d]^f \\
 B \times_A A^\mathbb{I} \ar[rr]_{A^{\partial_1}\pi_2} & & A }
\end{array}
\end{displaymath}

\begin{prop}{propofstrongheq}
The following are equivalent for a morphism $f: B \to A$:
\begin{enumerate}
\item[(i)] $f$ is a strong homotopy equivalence.
\item[(ii)] The map $\theta_f: f \to \partial_0 \hat\otimes f$ has a retraction in $\ct{E}^\to$.
\item[(iii)] The map $\sigma_f: \hat\exp(\partial_0,f) \to f$ has a section in $\ct{E}^\to$.
\end{enumerate}
\end{prop}
\begin{proof}
If $f$ is a strong homotopy equivalence, then $([g,H], K)$ is a retraction of $\theta_f$ as in
\diag{ B \ar[d]_f \ar[rr]^{i_1(\partial_1 \times B)} & & A \cup_B \mathbb{I} \times B \ar[d]_{\partial_0 \hat\otimes f} \ar[rr]^{[g, H]} & & B \ar[d]^f \\
A \ar[rr]_{\partial_1 \times A} & & \mathbb{I} \times A \ar[rr]_K & & A, }
and any retraction of $\theta_f$ must be of the form $([g,H],K)$ with $g, H$ and $K$ showing that $f$ is a strong homotopy equivalence.

Similarly, if $f$ is a strong homotopy equivalence, then $(\overline{H},(g,\overline{K})$ is a section of $\sigma_f$ as in
\diag{ B \ar[d]_f \ar[rr]^{\overline{H}} & & B^\mathbb{I} \ar[d]_{\hat\exp(\partial_0,f)} \ar[rr]^{B^{\partial_1}} & & B \ar[d]^f \\
A \ar[rr]_{(g, \overline{K})} & & B \times_A A^\mathbb{I} \ar[rr]_{A^{\partial_1}\pi_2} & & A, }
and any section of $\sigma_f$ must be of the form $(\overline{H},(g,\overline{K}))$ with $g, H$ and $K$ showing that $f$ is a strong homotopy equivalence.
\end{proof}

The equivalence of (i) and (ii) in the previous proposition is Lemma 3.3 in \cite{UniformFibrations}, while the next is Lemma 3.4 in the same source.

\begin{prop}{somestrhomeq}
The map $\partial_0 \hat{\otimes} f$ is a strong homotopy equivalence for any morphism $f: B \to A$.
\end{prop}
\begin{proof}
If $f: B \to A$ is any morphism, then  $\partial_0 \hat{\otimes} f$ is of the form $C \to \mathbb{I} \times A$ with $C = A \cup_B \mathbb{I} \times B$. Then $i_0\pi_A: \mathbb{I} \times A \to C$ is a homotopy inverse, as witnessed by
\[ H = [i_0\pi_A, i_1(\land \times B)]: \mathbb{I} \times A \cup_{\mathbb{I} \times B} (\mathbb{I} \times \mathbb{I} \times B) \cong \mathbb{I} \times C \to C \]
and
\[ K = \land \times A: \mathbb{I} \times \mathbb{I} \times A \to \mathbb{I} \times A. \]
In addition, the square
\diag{ \mathbb{I} \times C \ar[r]^H \ar[d]_{\mathbb{I} \times (\partial_0 \hat{\otimes} f)} & C \ar[d]^{\partial_0 \hat{\otimes} f} \\
\mathbb{I} \times \mathbb{I} \times A \ar[r]_K & \mathbb{I} \times A }
commutes, so the homotopy equivalence $f$ is strong.
\end{proof}

\begin{prop}{strhomeqstable}
Strong homotopy equivalences are stable under pullback along fibrations.
\end{prop}
\begin{proof}
This is Lemma 3.7 in \cite{UniformFibrations}.
\end{proof}

\begin{prop}{vogtslemmaforfibrations}
If $f: B \to A$ is a fibration and a homotopy equivalence between fibrant objects, then $f$ is a strong homotopy equivalence.
\end{prop}
\begin{proof}
Suppose $f: B \to A$ is a fibration and a homotopy equivalence between fibrant objects. This means that there is a homotopy inverse $g': A \to B$ and there are homotopies $H: g'f \simeq 1_B$ and $K: fg' \simeq 1_A$. Therefore
\diag{ \{ 0 \} \times A   \ar[r]^{g'} \ar[d]_{\partial_0 \hat\otimes \bot_A} & B \ar[d]^f \\
\mathbb{I} \times A \ar[r]_K \ar@{.>}[ur]^L & A }
commutes and because $f$ is a fibration, there will be a diagonal filler $L$. Writing $g = L(\partial_1 \times A)$, we see that $g$ is a section of $f$ with $g \simeq g'$. Hence $\pi_A: \mathbb{I} \times A \to A$ is a homotopy $fg \simeq 1$ and because $gf \simeq g'f \simeq 1_B$, there is a homotopy $M: gf \simeq 1_B$ as well. Our aim is to modify this homotopy $M$ to a homotopy $N$ making
\diag{ \mathbb{I} \times B \ar[d]_{\mathbb{I} \times f} \ar[r]^N & B \ar[d]^f \\
\mathbb{I} \times A \ar[r]_{\pi_A} & A }
commute.

For this we use the connections and the fact that
\diag{ \{ 0 \} \times \mathbb{I} \cup \{ 1 \} \times \mathbb{I} \cup \mathbb{I} \times \{ 0 \} \ar[rr]^(.7){[\overline{gfM}, \overline{\pi_B} , \overline{M}]} \ar[d]_{[\partial_0,\partial_1] \hat\otimes \partial_0} & & B^B \ar[d]^{f^B} \\
\mathbb{I} \times \mathbb{I} \ar[r]_\lor \ar@{.}[rru]^F & \mathbb{I} \ar[r]_{\overline{fM}} & A^B }
commutes. Since $f^B = \hat\exp(\bot_B: 0 \to B, f: B \to A)$ is a fibration by \refprop{onusingleibnizadjunction}.1, we obtain a diagonal filler $F: \mathbb{I} \times \mathbb{I} \to B^B$. Then $N = \overline{F(\mathbb{I} \times \partial_1)}$ is the desired homotopy.
\end{proof}

\begin{prop}{vogtslemmaforcofibrations}
If $f: B \to A$ is a cofibration and a homotopy equivalence between fibrant objects, then $f$ is a strong homotopy equivalence.
\end{prop}
\begin{proof}
The proof of this proposition is very similar to the previous one. Suppose $f: B \to A$ is a cofibration and a homotopy equivalence between fibrant objects. This means that there is a homotopy inverse $g': A \to B$ and there are homotopies $H: g'f \simeq 1_B$ and $K: fg' \simeq 1_A$. Therefore
\diag{ \{ 0 \} \times A \cup_B \mathbb{I} \times B \ar[rr]^(.6){[g', H]} \ar[d]_{\partial_0 \hat\otimes f} & & B\ar[d] \\
\mathbb{I} \times A \ar[rr] \ar@{.>}[urr]^L & & 1 }
commutes and because $f$ is a fibration, there will be a diagonal filler $L$. Writing $g = L(\partial_1 \times A)$, we see that $gf = 1_B$ and $g \simeq g'$. Hence $\pi_B: \mathbb{I} \times B \to B$ is a homotopy $gf \simeq 1$ and because $fg \simeq fg' \simeq 1_A$, there is a homotopy $M: fg \simeq 1_A$ as well. Our aim is to modify this homotopy $M$ to a homotopy $N$ making
\diag{ \mathbb{I} \times B \ar[d]_{\mathbb{I} \times f} \ar[r]^{\pi_B} & B \ar[d]^f \\
\mathbb{I} \times A \ar[r]_N & A }
commute.

For this we use the connections and the fact that
\diag{ \{ 0 \} \times \mathbb{I} \cup \{ 1 \} \times \mathbb{I} \cup \mathbb{I} \times \{ 0 \} \ar[rr]^(.7){[\overline{Mfg}, \overline{\pi_A}, \overline{M}]} \ar[d]_{[\partial_0,\partial_1] \hat\otimes \partial_0} & & A^A \ar[d]^{A^f} \\
\mathbb{I} \times \mathbb{I} \ar[r]_\lor \ar@{.}[rru]^F & \mathbb{I} \ar[r]_{\overline{Mf}} & A^B }
commutes. Since $A^f = \hat\exp(f: B \to A, !_A: A \to 1)$ is a fibration by \refprop{onusingleibnizadjunction}.1, we obtain a diagonal filler $F: \mathbb{I} \times \mathbb{I} \to B^B$. Then $N = \overline{F(\mathbb{I} \times \partial_1)}$ is the desired homotopy.
\end{proof}

\section{A model structure}
\label{sec:qms}

We continue working in the setting of the previous section and we establish the existence of a model structure on the full subcategory of fibrant objects. In addition, we establish that the resulting model structure gives a model of type theory with $\Pi$-types satisfying function extensionality.

\subsection{A WFS with cofibrations and trivial fibrations}

\begin{prop}{ontrivfibrations}
A morphism $f: B \to A$ is a trivial fibration if and only if it is a fibration and a strong homotopy equivalence.
\end{prop}
\begin{proof}
Suppose $f: B \to A$ is a trivial fibration. Because $0 \to A$ is a cofibration, the square
\diag{ 0 \ar[r] \ar[d] & B \ar[d]^f \\
A \ar@{.>}[ur]^g \ar[r]_1 & A }
has a diagonal filler $g$. Therefore $f$ has a section $g$ and $\pi_A: \mathbb{I} \times A \to A$ is a homotopy showing $fg \simeq 1$. Moreover, $[\partial_0,\partial_1] \times B = [\partial_0,\partial_1] \hat\otimes (0 \to B)$ is a cofibration as well, so also
\diag{ \{ 0 \} \times B + \{ 1 \} \times B \ar[rr]^{[gf,1_B]} \ar[d]_{[\partial_0,\partial_1] \times B} & & B \ar[d]^f \\
\mathbb{I} \times B \ar[r]_{\mathbb{I} \times f} \ar@{.>}[rru]^H & \mathbb{I} \times A \ar[r]_{\pi_A} & A }
has a diagonal filler, showing that $f$ is a strong homotopy equivalence.

Conversely, suppose $f$ is both a fibration and a strong homotopy equivalence. To show that $f$ is a trivial fibration, let $u$ be an arbitrary cofibration. Our goal is to show that $u \pitchfork f$. But since $f$ is a strong homotopy equivalence, it is a retract of $\hat\exp(\partial_0,f)$ by \refprop{propofstrongheq} and therefore it suffices to show that $u \pitchfork \hat\exp(\partial_0, f)$. But this follows from \refprop{Leibnizadjunction} and the fact that $f$ is a fibration.
\end{proof}

\subsection{A WFS with trivial cofibrations and fibrations}
\label{sec:wfs_triv_cof_fib}
\begin{prop}{ontrivcofibr}
If $u$ is a cofibration and a strong homotopy equivalence and $f$ is a fibration, then $u \pitchfork f$.
\end{prop}
\begin{proof} (This is Lemma 3.5.(ii) in \cite{UniformFibrations}.)
If $u$ is a strong homotopy equivalence, then $u$ is a retract of $u \hat\otimes \partial_0$ by \refprop{propofstrongheq}. So in order to show that $u \pitchfork f$ it suffices to show that $u \hat\otimes \partial_0 \pitchfork f$. But that is immediate from the fact that $u$ is a cofibration and $f$ is a fibration.
\end{proof}

\begin{prop}{factorisationtrivcoffibr}
Every morphism $f: B \to A$ between fibrant objects factors as a map which is both a cofibration and a homotopy equivalence followed by a fibration.
\end{prop}
\begin{proof}
Construct the following diagram, in which the square is a pullback:
\diag{ B \ar[dr]_{w} \ar@/^1pc/[rrd]^{\overline{\pi_A}f} \ar@/_1pc/[rdd]_{1_B} \\
& P_f \ar[d]_{p_1} \ar[r]^{p_2} & A^\mathbb{I} \ar[d]^{A^{\partial_0}} \\
& B \ar[r]_f & A. }
Since $A^{\partial_0} = \hat\exp(\partial_0,A \to 1)$ it follows from \refprop{onusingleibnizadjunction}.2 that this map is a trivial fibration. Since trivial fibrations are stable under pullback, $p_1$ is a trivial fibration and hence a homotopy equivalence. Since $1_B$ is a homotopy equivalence as well, it follows that $w$ is a homotopy equivalence.

Next, consider the map $p:= A^{\partial_1}p_2$. The square
\diag{ P_f \ar[d]_{(p_1,p)} \ar[r]^{p_2} & A^\mathbb{I} \ar[d]^{(A^{\partial_0}, A^{\partial_1})} \\
B \times A \ar[r]_{(f, 1)} & A \times A }
is a pullback and because $(A^{\partial_0}, A^{\partial_1}) = \hat\exp([\partial_0,\partial_1], A \to 1)$ is a fibration by \refprop{onusingleibnizadjunction}.1 and fibrations are stable under pullback, it follows that $(p_1, p)$ is a fibration as well. In addition, $B$ is fibrant, so $\pi_A: B \times A \to A$ is fibration and therefore $p = \pi_A(p_1,p)$ is
as well.

So $f = pw$ factors $f$ as a homotopy equivalence $w$ followed by a fibration $p$. Using the factorisation system that we have already established, we can write $w = w_1w_0$ where $w_1$ is a trivial fibration and $w_0$ is a cofibration. So $pw_1$ is a fibration, while $w_0$ is a homotopy equivalence, since both $w$ and $w_1$ are. Therefore $f = (pw_1)w_0$ factors $f$ as a cofibration which is also a homotopy equivalence followed by a fibration, as desired.
\end{proof}

Putting all the pieces together we can show:
\begin{theo}{modelstructure}
Let \ct{E} be an elementary topos with an interval object $\mathbb{I}$ and a class of cofibrations $\mathcal C$ satisfying the conditions at the start of Section 2. Then the full subcategory of \ct{E} on the fibrant objects carries a Quillen model structure in which the morphisms in $\mathcal C$ are the cofibrations, the fibrations as defined in \refdefi{fibrations} are the fibrations and the homotopy equivalences are the weak equivalences.
\end{theo}
\begin{proof}
  Weak equivalences satisfy the 2-out-of-3 condition by
  \refcoro{homeqsatisfies2out3}.

  By \refprop{WFSfromdominance}, (cofibrations, trivial fibrations) form a weak
  factorisation system. By \refprop{vogtslemmaforfibrations} and
  \refprop{ontrivfibrations} trivial fibrations are exactly fibrations that are weak
  equivalences. Hence, (cofibrations, acyclic fibrations) is a weak
  factorisation system.

  To show that (acyclic cofibrations, fibrations) form a weak
  factorisation system we use the retract argument
  (\reflemm{retractargument}). The factorisation is given by
  \refprop{factorisationtrivcoffibr}. If $u$ is a cofibration and a
  homotopy equivalence, and $f$ is a fibration, then $u \pitchfork f$
  by \refprop{vogtslemmaforcofibrations} and \refprop{ontrivcofibr}.
  The fibrations are closed under retracts because they are defined in
  terms of a lifting property, the cofibrations are closed under
  retracts because ${\mathcal C}$ is a dominance, and homotopy
  equivalences are closed under retracts by \refprop{htpyeq_closed_under_retracts}.
\end{proof}

\subsection{$\Pi$-types}
For the purpose of interpreting type theory in $\mathcal{E}_f$ we require $\Pi$- and $\Sigma$-types. The interpretation of $\Sigma$-types is trivial, as $\Sigma_f$ is just composition with $f$, and fibrations are stable under composition.

To interpret $\Pi$-types, we have to be a bit careful. A standard construction \cite{Seely:LCCC} allows us to leverage locally cartesian closed structure of a category to interpret $\Pi$-types.
Despite the fact that $\mathcal{E}$ is a topos, and hence is locally cartesian closed, we do not necessary know that $\mathcal{E}_f$ is locally cartesian closed.
However, for the purposes of interpreting type theory, we do not need all adjunctions $\Sigma_f \dashv f^\ast \dashv \Pi_f$ to be present in $\mathcal{E}_f$; we only require $\Pi_f(g)$ to exist in $\mathcal{E}_f$ whenever $f$ and $g$ are fibrations,
that is, we require an adjunction $f^\ast : (\mathcal{E}/A)_f \to (\mathcal{E}/B)_f : \Pi_f$ for a fibration $f : B \to A$ between fibrant objects.
This follows from the following proposition:

\begin{prop}{Pitypesexist}
For any fibration $f: B \to A$ the right adjoint $\Pi_f: \ct{E}/B \to \ct{E}/A$ to pullback preserves fibrations.
\end{prop}
\begin{proof}
Suppose $g: C \to B$ is a fibration and $\Pi_f(g)$ fits into a commuting square of the form
\diag{ \bullet \ar[d]_{\partial_0 \hat{\otimes} u} \ar[r] & \bullet \ar[d]^{\Pi_f(g)} \\
\bullet \ar[r] & A }
in which $u$ is a cofibration. The aim is to show that this square has a diagonal filler. By the adjunction $f^* \ladj \Pi_f$ this is the same as showing that
\diag{ \bullet \ar[d]_{f^*(\partial_0 \hat{\otimes} u)} \ar[r] & \bullet \ar[d]^{g} \\
\bullet \ar[r] & B }
has a diagonal filler. Since $\partial_0 \hat{\otimes} u$ is a cofibration and a strong homotopy equivalence by \refprop{somestrhomeq}, the same is true for $f^*(\partial_0 \hat{\otimes} u)$ by \refprop{strhomeqstable}. Therefore the second square has a diagonal filler by \refprop{ontrivcofibr}.
\end{proof}

To show that we do not just have $\Pi$-types, but that they also satisfy function extensionality we show the following proposition which implies this principle by Lemma 5.9 in \cite{shulman13a}.

\begin{prop}{functionextensionality}
For any fibration $f: B \to A$ the right adjoint $\Pi_f: \ct{E}/B \to \ct{E}/A$ to pullback preserves trivial fibrations.
\end{prop}
\begin{proof}
Cofibrations are stable under pullback along fibrations (in fact, along any map), so a similar argument as in the previous proposition establishes that trivial fibrations are stable under $\Pi$ along fibrations.
\end{proof}


\renewcommand{\id}{1}
\section{The effective topos}
\label{sec:eff}
For the remainder of this paper we work with the \emph{effective topos} $\Eff$.
We briefly describe the effective topos and the category of assemblies, without giving any proofs.
An interested reader is referred to a comprehensive book \cite{vanOosten:realiz}, the lecture notes \cite{Streicher:REAL}, and the original paper \cite{H:eff} on the subject.
We frequently conflate recursive functions and their G\"odel codes, and we use standard notation $a \cdot b$ for Kleene application and standard notation $\lambda \langle x, y \rangle.t$ for pattern-matching in $\lambda$-functions.

The objects of $\Eff$ are pairs $(X, \sim)$ where $X$ is a set and $\sim$ is a ${\mathcal P}(\omega)$-indexed partial equivalence relation on $X$; that is $\sim$ is a mapping $X \times X \to {\mathcal P}(\omega)$.
We denote $\sim(x,y)$ by $[x \sim y]$.
We require the existence of computable functions $\sym$ and $\tr$, such that if $n \in [x \sim y]$, then $\sym(n) \in [y \sim x]$ and if $m \in [y \sim z]$, then $\tr(n, m) \in [x \sim z]$.

A morphism $F : (X, \sim) \to (Y, \approx)$ is a ${\mathcal P}(\omega)$-indexed functional relation between $X$ and $Y$ that respects $\sim$ and $\approx$.
Specifically, $F$ is a mapping $X \times Y \to {\mathcal P}(\omega)$ and we require the existence of computable functions $\st_X$, $\st_Y$, $\rel$, $\sv$ and $\tot$ satisfying

\begin{itemize}
\item If $n \in F(x, y)$, then $\st_X(n) \in [x \sim x]$ and $\st_Y(n) \in [y \approx y]$;
\item If $n \in F(x, y)$ and $k \in [x \sim x']$ and
  $l \in [y \approx y']$, then $\rel(n,k,l) \in F(x', y')$;
\item If $n \in F(x,y)$ and $m \in F(x, y')$, then $\sv(n,m) \in [y \approx y']$;
\item If $n \in [x \sim x]$, then $\tot(n) \in \bigcup_{y \in Y} F(x,y)$.
\end{itemize}

Two functional relations $F, G: X \times Y \to \mathcal{P}(\omega)$ are said to be equal if there is a computable function $\varphi$ such that if $n \in F(x,y)$, then $\varphi(n) \in G(x,y)$.
The identity arrow on $(X, \sim)$ is represented by the relation $\sim$ itself.

Given two sets $A, B \in {\mathcal P}(\omega)$ we write $A \wedge B$ for the set $\{ \langle a, b \rangle \mid a \in A, b \in B \}$ where $\langle a, b\rangle$ is a surjective pairing of $a$ and $b$.
Then the composition $G \circ F$ of two functional relations $F : (X, \sim) \to (Y, \approx)$ and $G : (Y, \approx) \to (Z, \approxeq)$ is defined as $(G \circ F)(x, z) = \bigcup_{y \in Y} F(x, y) \wedge G(y, z)$.


\paragraph{Constant objects functor.}
The internal logic of $\Eff$, as is the case with any topos, has the so-called \emph{local operator} $\neg\neg : \Omega \to \Omega$.
Given an object $(A, \sim)$ and a subobject $(A', \sim_{A'})$, the latter is said to be $\neg\neg$-dense in $(A, \sim)$ if $\forall a:A (\neg\neg(A'(a)))$ holds; that is, if $A'(x)$ is non-empty whenever $[x \sim x]$ is non-empty.
An object $X$ is said to be a $\neg\neg$-sheaf if for any dense $A' \hookrightarrow A$ any map $A' \to X$ can be extended to a map $A \to X$.
In the effective topos the $\neg\neg$-sheaves can be described as objects in the image of a ``constant object functor'' $\nabla$.

\begin{definition}
  The functor $\nabla : \Sets \to \Eff$ is defined on objects as
  $\nabla(X) = (X, \sim)$ where
  \[
    [x \sim x'] =
    \begin{cases}
      \omega & \mbox{ if } x = x' \\
      \emptyset & \mbox{ otherwise}
    \end{cases}
  \]
  and on morphisms as
  \[
    \nabla(f : X \to Y)(x, y) = [x \sim f(y)]
  \]
\end{definition}

The functor $\nabla$, together with the global sections functor $\Gamma(X) = Hom_{\Eff}(1, X)$, forms a geometric morphism $\Gamma \dashv \nabla$ which embeds $\Sets$ into $\Eff$.
Note that in particular $\Gamma$ preserves finite limits and arbitrary colimits (including preservation of monomorphisms and epimorphisms) and $\nabla$ preserves arbitrary limits.

\subsection{Assemblies.} 
We say that an object $A$ is $\neg\neg$-separated if $\forall x:A \forall y:A (\neg\neg(x \sim y) \to (x \sim y))$; that is if we know that $[x \sim y]$ is non-empty and $n \in [x \sim x], m \in [y \sim y]$, then we can recursively find $\phi(n, m) \in [x \sim y]$.
Just like $\neg\neg$-sheaves are objects in the image of the inclusion of $\Sets$, the $\neg\neg$-separated objects can be described as objects in the image of the inclusion of the category of \emph{assemblies} into $\Eff$.

\begin{definition}
  An \emph{assembly} is a pair $(X, E_X)$ where $X$ is a set, and $E_X : X \to {\mathcal P}(\omega)$ is a function such that $E_X(x) \neq \emptyset$ for every $x \in X$.
  We will call such a function a \emph{realizability relation} on $X$.

  A morphism of assemblies $f : (X, E_X) \to (Y, E_Y)$ is a map $f : X \to Y$ such that there is a computable function $\varphi$ such that for every $x \in X$ and $n \in E_X(x)$, $\varphi(n) \downarrow$ and $\varphi(n) \in E_Y(f(x))$.
  In this case we say that $\varphi$ \emph{tracks} or \emph{realizes} $f$.
\end{definition}

We denote the category of assemblies and assembly morphisms as $\Asm$.
Sometimes we drop the realizability relation if it is obvious from the context.
We also write $n \Vdash_X x$ for $n \in E_X(x)$.

\begin{example}
  The natural numbers object $\mathbf{N}$ in $\Eff$ is an assembly $(\omega, E_{\mathbf{N}})$ with $E_{\mathbf{N}}(i) = \{ i \}$.
\end{example}

\begin{example}
  The terminal object $1$ of $\Eff$ is an assembly $(\{ \ast \}, E_1)$ with $E_1(\ast) = \{ 0 \}$.
\end{example}


The category of assemblies is a full subcategory of the effective topos via an inclusion which sends an assembly $(X, E_X)$ to an object $(X, \sim_X)$ where
\[
  [x \sim_X x'] =
  \begin{cases}
    E_X(x) & \mbox{ if } x = x' \\
    \emptyset & \mbox{ otherwise }
  \end{cases}
\]
and which sends a map $f : (X, E_X) \to (Y, E_Y)$ to an induced relation
\[
  F(x, y) = [x \sim_X x] \wedge [y \sim_Y f(x)]
\]

\begin{example}
  Note that every $\nabla(X)$ is an assembly $(X, E)$ with $E(x) = \omega$, and $\nabla$ factors through $\Asm \hookrightarrow \Eff$.
\end{example}

\subsection{Model structure on $\Eff_f$}
\label{sec:model_structure_eff}
In order to apply the result from Section \ref{sec:qms} to the effective topos $\Eff$, we must select an interval object $\mathbb{I}$ and a class of morphisms ${\mathcal C}$ satisfying certain conditions.
We take $\mathcal{C}$ to be $Mon$, the class of all monomorphisms, and we take the interval object to be $\mathbb{I} = \nabla(2)$.
Alternatively, $\mathbb{I}$ can be described as an assembly $(\{0,1\},E)$ with $E(i) = \{ 1 \}$.
The connection structure $\land, \lor : \mathbb{I} \times \mathbb{I} \to \mathbb{I}$ is defined simply as
\begin{align*}
  x \land y & = min(x, y) \qquad \mbox{tracked by } \lambda x.0 \\
  x \lor y & = max(x, y) \qquad \mbox{tracked by } \lambda x.0 
\end{align*}
It is straightforward to verify that the class of monomorphisms satisfies the conditions outlined at the beginning of Section \ref{sec:setup}.

For the rest of this paper we use the following notation.
We write $s$ and $t$ for the source and target maps $X^{\partial_0} : X^{\mathbb{I}} \to X$ and $X^{\partial_1} : X^{\mathbb{I}} \to X$, respectfully.
We write $r$ for the ``reflexivity'' map $X^{!_\mathbb{I}} : X \to X^{\mathbb{I}}$.

\section{Contractible maps in $\Eff$}
\label{sec:eff_contr}
In this section we are going to characterize contractible objects in $\Eff$ as uniform inhabited objects (\refprop{fibrant_uniform_impl_contractible}), and characterize trivial fibrations in $\Asm$ as uniform epimorphisms (\refprop{uniform_maps_asm}).
The latter characterization will allow us to give a concrete description of fibrant assemblies in terms of realizers (\refprop{char_fibrant_asm}).

\subsection{Uniform objects and contractibility}

\begin{defi}{uniform}
  An object $(X, \sim)$ is said to be \emph{uniform} if it is covered by a $\neg\neg$-sheaf, i.e.
  there is an epimorphism $\nabla Y \to (X, \sim)$.
\end{defi}

\begin{prop}{uniform_characterisation}
  An object is uniform if it is isomorphic to an object $(X, \sim)$ such that there is a number $n \in \bigcap_{x:X} [x \sim x]$.
\end{prop}
\begin{proof}
  By  \cite[Proposition 2.4.6]{vanOosten:realiz}.
\end{proof}

Recall than an object $X$ is said to be \emph{contractible} if the unique map $X \to 1$ is a trivial fibration.
In our case, since the dominance $\mathcal{C}$ is exactly the class of monomorphisms, contractible objects are exactly the \emph{injective} objects. 
In an elementary topos every injective object is a retract of some power-object.
Because every power-object is \emph{uniform} in the effective topos, it follows that every contractible object must be uniform.
It is natural to ask if the converse of this fact holds as well.
The answer to this question is ``no'', unless we restrict ourselves to fibrant objects.

\begin{prop}{contractible_impl_uniform}
  Every contractible (injective) object is uniform and has a global element. 
\end{prop}
\begin{proof}
  Suppose $X$ is an injective object in $\Eff$.
  By a topos-theoretic argument, $X$ is a retract of ${\mathcal P}(X)$.
  It has been shown in \cite[Proposition 3.2.6]{vanOosten:realiz} that every powerset is uniform.
  Because $X$ is covered by a uniform object, we can conclude that $X$ is uniform itself.

  A global element of $X$ can be obtained by extending the unique map $0 \to X$ along the monomorphism $0 \to 1$.
\end{proof}

\begin{prop}{fibrant_uniform_impl_contractible}
  If a uniform fibrant object $(X, \sim)$ has a global element $s : 1 \to (X, \sim)$, then $(X, \sim)$ is contractible.
\end{prop}
\begin{proof}
  We can assume that $s$ is of the form $s(\ast, x) = [x \sim c]$ for some $c \in X$.
  We shall prove that $s$ is a homotopy equivalence with homotopy inverse $!_X : (X, \sim) \to 1$.

  The composition $!_X \circ s$ is the identity by the universal property of the terminal object.
  The homotopy $\theta : s \circ !_X \sim 1_X$ is constructed as follows.
\[
    \begin{cases}
      \theta(0, x, y) = s(\ast, y) = [y \sim c] \\
      \theta(1, x, y) = [x \sim y]
    \end{cases}
\]
Clearly, $\theta : \mathbb{I} \times (X, \sim) \to (X, \sim)$ is strict and single-valued.
To see that $\theta$ is total, it suffices to provide an element $\psi(n) \in \theta(0, x, y_0) \cap \theta(1, x, y_1) = [y_0 \sim c] \cap [x \sim y_1]$ for some $y_0, y_1$ given that $n \in [x \sim x]$.
But if we take $y_0 = c$ and $y_1 = x$, then the required element $\psi(n) \in [c \sim c] \cap [x \sim x]$ can be obtained from the uniformity of $(X, \sim)$ by \refprop{uniform_characterisation}.
\end{proof}

\subsection{Uniform maps and fibrant assemblies}
In the previous subsection we have discussed uniform objects.
Now we move on to uniform maps.

\begin{defi}{uniform_map}
  A map $F : (Y, \approx) \to (X, \sim)$ is said to be \emph{uniform} if it is covered by a $\neg\neg$-sheaf in the slice topos $\Eff/(X, \sim)$.
  That is, there is a map $\alpha : Z \to \Gamma(X, \sim)$ such that a map $S : (X, \sim) \times_{\nabla\Gamma(X, \sim)} \nabla(Z) \to (X, \sim)$ is a pullback of $\nabla(\alpha)$ along $\eta_X : (X, \sim) \to \nabla\Gamma(X, \sim)$, and there is an epimorphism $R : (X, \sim) \times_{\nabla\Gamma(X, \sim)} \nabla(Z) \to (Y, \approx)$ over $(X, \sim)$, as depicted below.
  \[
    \xymatrix{
      (X, \sim) \times_{\nabla\Gamma(X, \sim)} \nabla(Z) \ar[r] \ar[d]_S & \nabla(Z) \ar[d]^{\nabla(\alpha)}\\
      (X, \sim) \ar[r]_{\eta_X}& \nabla\Gamma(X, \sim)
    }
    \qquad
    \xymatrix{
      (X, \sim) \times_{\nabla\Gamma(X, \sim)} \nabla(Z) \ar@{->>}[r]^-{R} \ar[dr]_S & (Y, \approx) \ar[d]^{F} \\
      & (X, \sim)
      }
  \]
\end{defi}

\begin{prop}{uniform_map_characterisation}
  A map $F : (Y, \sim) \to (X, \approx)$ is uniform iff there are recursive functions $\alpha, \beta$ such that for all $y \in Y$, $x \in X$, $n \in [x \approx x]$, $m \in F(y, x)$ there exists an $y' \in Y$ and
\[
  \begin{cases}
    \alpha(n) \in F(y', x) \\
    \beta(n, m) \in [y \sim y']
  \end{cases}
\]
In particular a map $f : (Y, E_Y) \to (X, E_X)$ between assemblies is uniform iff there is a recursive $\alpha$ such that
  \[
    \forall x \in X \forall y \in Y \forall n \in E_X(x) (f(y) = x \to \alpha(n) \in E_Y(y))
  \]
  In other words, $\alpha(n) \in \bigcap_{y \in f^{-1}(x)} (E_Y(y))$ whenever $n \in E_X(x)$.
  In such a situation we say that every fiber of $f$ is uniform and $\alpha$ witnesses the uniformity.
\end{prop}
\begin{proof}
  By  \cite[Proposition 3.4.6]{vanOosten:realiz}.
\end{proof}

The next proposition is aimed at generalizing \refprop{contractible_impl_uniform} to uniform maps.
We have not managed to extend the correspondence to arbitrary uniform maps.
However, we can generalize the correspondence to the uniform maps in $\Asm$ (\refprop{uniform_maps_asm}).

\begin{theo}{uniform_maps_negnegsep}
  Let $F : (Y, \approx) \to (X, \sim)$ be a map and let $(Y, \approx)$ be $\neg\neg$-separated.
  If $F$ is a trivial fibration, then $F$ is a uniform map.
\end{theo}
\begin{proof}
  Consider the following pullback
  \begin{displaymath}
    \xymatrix{
      (A, \asymp) \ar[r] \ar[d]_{\pi} & \nabla \Gamma(Y, \approx) \ar[d]^{\nabla \Gamma F}\\
      (X, \sim) \ar[r]_{\eta} & \nabla \Gamma(X, \sim)
      }
  \end{displaymath}
  The object $(A, \asymp)$ can be described as
  \[
  A = \{ ([y], x) \mid \nabla\Gamma(F)([y]) = [x] \}
  \]
  where $[y]$ is the equivalence class of $y'$ such that $[y \approx y']$ is non-empty, thus $\nabla\Gamma(F)([y]) = [x]$ means that $F(y, x)$ is non-empty; the realizability relation on $A$ is
  \[
  ([y], x) \asymp ([y'], x') =
  \begin{cases}
    [x \sim x'] & \mbox{ if } [y] = [y'] \mbox{ i.e. } [y \approx y'] \neq \emptyset \\
    \emptyset & \mbox{ otherwise }
  \end{cases}
  \]
  Then consider a map $S : (Y, \sim) \to (A, \asymp)$ defined as $S = \langle F, \eta_Y \rangle$.
  Explicitly:
  \[
    S(y, [y'], x) = F(y, x) \wedge \{ 0 \mid y \in [y'] \mbox{ i.e. } [y \approx y'] \neq \emptyset \}
  \]
  If $(Y, \approx)$ is $\neg\neg$-separated, then $S$ is a mono.
  To see this, suppose $\langle m_1, m_2 \rangle \in S(y_1, [y], x) \wedge S(y_2, [y], x)$.
  We are to provide an element of $[y_1 \approx y_2]$.
  Because $m_1 \in S(y_1, [y], x)$ we know that $[y \approx y_1]$ is non-empty.
  Similarly for $m_2$ and $y_2$.
  Then, from $m_1$ and $m_2$ we can get realizers for $[y_1 \approx y_1]$ and $[y_2 \approx y_2]$.
  Then $[y_1 \approx y_2]$ follows from $\neg\neg$-separation.

  Then, because $S$ is a mono and $F$ is a trivial fibration, the square below has a filler $H : (A, \asymp) \to (Y, \sim)$.
  \begin{displaymath}
    \xymatrix{
      (Y, \approx) \ar[d]_S \ar@{=}[r] & (Y, \approx) \ar[d]^F \\
      (A, \asymp) \ar@{-->}[ur]^{H} \ar[r]_{\pi} & (X, \sim)
      }
    \end{displaymath}
    Then $H$ is an epimorphism, as it is a retract, and hence $F$ is a uniform map.
    
\end{proof}

Now we can show:
\begin{prop}{uniform_maps_asm}
  A map $f$ is a trivial fibration between assemblies iff it is a uniform epimorphism between assemblies.
\end{prop}
\begin{proof} 
  $(\Leftarrow)$ For the ``if'' direction, suppose $f$ is a uniform epimorphism, with uniformity witnessed by $\alpha$ (in the sense of \Cref{prop:uniform_map_characterisation}), and we have the following commutative diagram in which $i$ is a monomorphism:
  \begin{displaymath}
    \xymatrix{ A \ar@{>->}[d]_{i} \ar[r]^{g} & Y \ar[d]^{f} \\
       B \ar[r]_{h} & X  }
  \end{displaymath}
  As $\Gamma$ preserves monomorphisms and epimorphisms of $\Asm$, we can find a filler $k : B \to Y$ for the diagram above in $\Sets$, under the image of $\Gamma$ (such filler exists by axiom of choice).
  Thus, to fill in the diagram above in $\Asm$ we are to find a realizer for $k$.
  One can check that the realizer is provided by $\lambda n.
  \alpha(\underline{h} \cdot n)$, where $\underline{h}$ is a realizer for $h$.

  $(\Rightarrow)$ The ``only if'' direction follows from \reftheo{uniform_maps_negnegsep}.
\end{proof}

Using the \refprop{uniform_maps_asm} we can characterize the fibrant assemblies recursive-theoretically in $\Eff$. For this, we need to introduce a notion of path-connectedness.

\begin{definition}
  Let $X$ be an assembly, and let $x \in X$.
  A \emph{path-connected component} of $x$, denoted as $[x]$ is a set of $y \in X$ such that there is a map $p : \mathbb{I} \to X$ such that $p(0) = x$ and $p(1) = y$.
  We also say that $y$ is path-connected to $x$.
\end{definition}

\begin{prop}{char_fibrant_asm}
  An assembly $X$ is fibrant iff for every $n \in E_X(x)$ one can uniformly find $\alpha(n)$ that realizes the path-connected component of $x$, i.e.
  $\alpha(n) \in \bigcap_{y \in [x]} E_X(y)$.
\end{prop}
\begin{proof}
  By \refprop{onusingleibnizadjunction}, an assembly $X$ is fibrant iff both $s = \hat\exp(\partial_0, X \to 1), t = \hat\exp(\partial_1, X \to 1): X^{\mathbb{I}} \to X$ are trivial fibrations.
  Note that the interval object $\nabla(2)$ comes with ``twist'' map $tw : \nabla(2) \to \nabla(2)$ which is a self-inverse and which satisfies $X^{tw} \circ s = t$, $X^{tw} \circ t = s$.
  Thus, for an assembly $X$ to be fibrant it is sufficient to check that the source map $s : X^\mathbb{I} \to X$ is a trivial fibration.
  Then apply \refprop{uniform_maps_asm} to the map $s : X^{\mathbb{I}} \to X$.
\end{proof}

\section{Discrete objects and discrete reflection}
\label{sec:eff_discr}
In this section we describe the reflexive subcategory of discrete objects in $\Eff$ and show that every discrete object is fibrant.
We also prove that the unit of the discrete reflection of a fibrant assembly is a homotopy equivalence, which allows us to concretely characterize the homotopy category of fibrant assemblies as the category of modest sets (\refprop{asm_ho_mod}).

\subsection{Discrete objects and discrete maps}

\begin{defi}{discrete}
  An object of $\Eff$ is said to be \emph{discrete} if it is a quotient of the subobject of the natural numbers object.
\end{defi}

The following proposition characterizes discrete objects up-to isomorphism.

\begin{prop}{discrete_characterisation}
  An object of $\Eff$ is discrete iff it is isomorphic to an object $(X, \sim)$ such that $n \in [x \sim x] \cap [y \sim y]$ implies that $x = y$.
\end{prop}
\begin{proof}
  By \cite[Proposition 3.2.20]{vanOosten:realiz}.
\end{proof}

Discrete objects can be characterized as objects which have no non-constant paths.

\begin{prop}{discrete_obj_no_nontriv_paths}
  An object $X$ is discrete if and only if the canonical map $p : X \to X^{\mathbb{I}}$ is an isomorphism.
  Hence, if $H : \mathbb{I} \times A \to X$ is a morphism into a discrete object $X$, then there is a map $h : A \to X$ such that $H = h \circ \pi_2$.
\end{prop}
\begin{proof}  
  The first point holds by \cite[Propositions 3.2.21 and 3.2.22]{vanOosten:realiz}.

  For the second point, take $h$ to be the composite $A \xrightarrow{\bar{H}} X^{\mathbb{I}} \to X$.
\end{proof}

Discreteness can be generalized from objects to maps as follows:

\begin{defi}{discrete_map}
  A map $F : (Y, \approx) \to (X, \sim)$ is \emph{discrete} if it is a quotient of the subobject of the natural numbers object in $\Eff/(X, \sim)$, which is represented by a map $(X, \sim) \times \N \to (X, \sim)$.
\end{defi}

\begin{proposition}[{\cite[Proposition 3.4.3]{vanOosten:realiz}}]
  If $F$ is a discrete map, then there is a computable function $\varphi$ that given $n \in F(y, x), m \in F(y', x)$, $u \in [y \approx y] \cap [y' \approx y']$ provides an element $\varphi(n, m,u) \in [y \approx y']$.
\end{proposition}

\begin{prop}{discrete_map_discr_base_fibrant}
  Every discrete map $F : (Y, \approx) \to (X, \sim)$ with a discrete base $(X, \sim)$ is a fibration.
\end{prop}
\begin{proof}
  By \refprop{discrete_characterisation} we may assume that $X$ is of the form $(X, \sim)$ such that $([x \sim x] \cap [x' \sim x']) \neq \emptyset \implies x = x'$.

  Because $F$ is a discrete map, the following proposition holds in $\Eff$:
  \[
    \forall y y' x. (F(y, x) \wedge F(y', x) \wedge ([y\approx y] \cap [y' \approx y']) \to [y\approx y'])
  \]
  Then consider the following lifting problem
  \begin{displaymath}
    \xymatrixcolsep{5pc}
    \xymatrix{ (\{0\} \times B) \cup (\mathbb{I} \times A) \ar@{>->}[d]_{\partial_0 \leibniz u} \ar[r]^-{[\alpha_0, \alpha_1]} & (Y, \approx) \ar[d] \\
       \mathbb{I} \times B \ar[r]_{\beta} & (X, \sim)  }
  \end{displaymath}
  Note that $\alpha_1(0, a, y) \simeq \alpha_1(1, a, y)$.
  For a given $n \in [a \sim a]$ we can obtain $\tot_{\alpha_1}(1, n) \in \bigcup_{y \in Y} \alpha_1(0, a, y) \cap \bigcup_{y\in Y} \alpha_1(1, a, y)$, i.e.
  $\tot_{\alpha_1}(1, n) \in \alpha_1(0, a, y_0) \cap \alpha_1(1, a, y_1)$ for some $y_0, y_1 \in Y$.
  Thus, $\st_Y(\tot_{\alpha_1}(1, n)) \in [y_0 \approx y_0] \cap [y_1 \approx y_1]$.
  If we can show that both $y_0$ and $y_1$ map to the same element under $F$, then, employing the property of the discrete maps, we can show that $[y_0 \approx y_1]$; then, if $m \in \alpha_1(0, a, y)$ we get $[y \approx y_0]$ by single-valuedness, and $[y \approx y_1]$ by transitivity, hence we get $\alpha_1(1, a, y)$ because $\alpha_1$ respects $\approx$.

  To see that $y_0$ and $y_1$ gets mapped to the same basepoint, we can use similar reasoning as in the previous point to establish that $\beta(0, b, x) \simeq \beta(1, b, x)$.
  By commutativity of the diagram, we have for any $x$,
  $F(y_0, x) \simeq (\beta \circ (I \times u))(0, a, x)$ and $F(y_1, x) \simeq (\beta \circ (I \times u))(1, a, x)$.
  However, $(\beta \circ (I \times u))(0, a, x) = \exists b.
  u(a, b) \wedge \beta(0, b, x) \simeq \exists b.
  u(a, b) \wedge \beta(1, b, x) = (\beta \circ (I \times u))(1, a, x)$.
  From that we get $F(y_0, x) \wedge F(y_1, x)$.
\end{proof}

\begin{coro}{discrete_impl_fibrant}
  Every discrete object $X$ is fibrant.
\end{coro}
\begin{proof}
  $X$ is discrete iff the map $X \to 1$ is discrete.
\end{proof}

\begin{example}
  Recall that a \emph{modest set} is a discrete assembly.
  Examples of modest sets are $\N$ and $1$; in fact, all finite types in $\Eff$ are modest sets.

  We denote the full subcategory of modest sets as $\Mod \hookrightarrow \Eff$.
  By \refcoro{discrete_impl_fibrant} every modest set is fibrant.
\end{example}

\refprop{discrete_map_discr_base_fibrant} cannot be extended to arbitrary discrete maps.
For this consider the following counterexample.

\begin{example}
\label{example:discrete_map_no_fib}
Note that if we restrict our attention to the category of assemblies, then a map $f : (Y, E_Y) \to (X, E_X)$ is discrete iff every fiber $f^{-1}(x)$ is discrete (\cite[Proposition 3.4.4]{vanOosten:realiz}).

Consider an inclusion of assemblies $f = [ \partial_0, \partial_1 ] : 2 \to \mathbb{I}$.
Each fiber $f^{-1}(i)$ is discrete, hence $f$ is a discrete map.
However, it is not a fibration.
Consider the following lifting problem (in the category of assemblies):
\[
  \xymatrixcolsep{5pc}
  \xymatrix{
    (\{0\} \times \mathbb{I}) \cup (\mathbb{I} \times \{ 0 \}) \ar[d]_{\partial_0 \leibniz \partial_0} \ar[r]^-{[\phi, \phi]} \ar[d]
      & 2 \ar[d]^{[\partial_0, \partial_1]}\\
    \mathbb{I} \times \mathbb{I} \ar[r]_{\lor} & \mathbb{I}
  }
\]
Here the map $(\{0\} \times \mathbb{I}) \cup (\mathbb{I} \times \{ 0 \}) \to \mathbb{I} \times \mathbb{I}$ embeds an open box without two sides $\llcorner$ into a square $\Box$, and the map $\phi$ 
is defined as $\phi(i) = 0$, and $\lor$ is a connection defined at the beginning of Section \ref{sec:model_structure_eff}.
We claim that this lifting problem has no solution; for suppose $h$ is such a filler.
Because $h : \mathbb{I} \times \mathbb{I} \to 2$ is a map from a uniform object into a discrete object, $h$ has to be constant.
Because $[\partial_0, \partial_1]$ is essentially an identity (on the level of sets), $[\partial_0, \partial_1](h(1, 1)) = [\partial_0, \partial_1](h(0, 0))$.
However, $1 \lor 1 = 1 \neq 0 = 0 \lor 0$, so the lower square cannot commute.
Hence, $f$ is not a fibration.
\end{example}

\subsection{Path contraction and discrete reflection}

The inclusion of discrete objects in the effective topos has a left adjoint called the \emph{discrete reflection}, see \cite[Proposition 3.2.19]{vanOosten:realiz}.
It was noted in \cite{htpyEff} that discrete reflection can be seen as internally as a set of path-connected components.

\begin{proposition}
The discrete reflection $X_d$ of an object $X$ is a coequalizer of the diagram
\[
  \xymatrix{
    X^{\mathbb{I}} \ar@<1ex>[r]^-{s} \ar@<-1ex>[r]_-{t} & X \ar@{->>}[r]^q & X_d
    }
\]
\end{proposition}
\begin{proof}
  First, we check that $X_d$ is discrete.
  For this, we reason in the internal logic.
  Let $\pi : \mathbb{I} \to X_d$ be a path.
  We will show that it is trivial, i.e. $\pi = \pi \circ \partial_0 \circ !_{\mathbb{I}}$.
  Because $\mathbb{I} = \nabla(2)$ is internally projective (\cite[Proposition 3.2.7]{vanOosten:realiz}), there is a map $p : \mathbb{I} \to X$ such that $q \circ p = \pi$.
  Define $P = \overline{p \circ \wedge} : \mathbb{I} \to X^{\mathbb{I}}$. Then $qsP = qtP$, as $q$ coequalizes $s$ and $t$. But $tP = p$ and $sP = p \circ \partial_0 \circ !_{\mathbb{I}}$. Hence, $qp = \pi = q p \partial_0 !_{\mathbb{I}} = \pi \circ \partial_0 \circ !_\mathbb{I}$. 

  Therefore, $X_d$ is discrete.
  To see that it satisfies the universal property, let $f : X \to D$ be a map into a discrete object $D$.
  Then $f \circ s = s \circ f^{\mathbb{I}}$ and $f \circ t = t \circ f^{\mathbb{I}}$, by the naturality.
  By \Cref{prop:discrete_obj_no_nontriv_paths}, $s = t : D^{\mathbb{I}} \to D$, hence $f \circ s = f \circ t$.
  As $q$ is the coequalizer of $s$ and $t$, there is a unique map $\bar{f}$ such that $f = \bar{f} \circ q$.
\end{proof}

It is known that in a model category where every object is cofibrant, every fibrant object can be equipped with a weak groupoid structure.
We will need the path composition operation of the groupoid for the characterization of discrete reflection. Specifically, there is a composition operation $c : X^{\mathbb{I}} \times_X X^{\mathbb{I}} \to X^{\mathbb{I}}$ satisfying
  \begin{itemize}
  \item $s \circ c = s \circ \pi_1$ and $t \circ c = t \circ \pi_2$;    
  \item $c  \langle rs, \id \rangle \sim \id$;
  \item $c \langle \id, rt \rangle \sim \id$;
  \item $c \langle c ,  \id \rangle \sim c \langle \id, c \rangle$.
  \end{itemize}
See, e.g., \cite[Appendix A.1]{benno:pathcat_id_types} for explicit constructions.

It follows, using the composition operation, that for a fibrant object $X$, the image of $\langle s, t \rangle : X^{\mathbb{I}} \to X \times X$ is an equivalence relation. Thus, for a fibrant assembly $(X, E_X)$ the discrete reflection $X_d$ can be described as an assembly $(X / \sim_p, E)$ where
$$
x \sim_p y \iff \exists p : \mathbb{I} \to X (p(0) = x \wedge p(1) = y)
$$
and $E([x]) = \bigcup_{y \in [x]} E_X(y)$. One can check directly that $X_d$ is indeed the discrete reflection of $X$ with the unit $\eta_X : x \mapsto [x]$ tracked by $\lambda x .x$.

Using this explicit description we can prove the following statement.

\begin{prop}{discr_reflection_htpy_eq}
  For a fibrant assembly $X$, the unit of the discrete reflection unit $\eta : X \to X_d$ is a homotopy equivalence.
\end{prop}
\begin{proof}
  Using axiom of choice, one can pick for each $x \in X$ a canonical representative $g([x]) \in [x]$ of each equivalence class $[x] \in X_d$.
  By \refprop{char_fibrant_asm}, there is a recursive $\alpha$ witnessing the uniformity of $s : X^{\mathbb{I}} \to X$.
  One can then verify that the following function tracks $g : X_d \to X$:
  \[
    \lambda n. \alpha (n) \cdot 1
  \]

  Clearly, $\eta \circ g = \id_{X_d}$.
  We are to show that there is a homotopy $g \circ \eta \sim \id_X$.
  Intuitively, this is the case because $g([x]) \in [x]$, and thus $g([x])$ must be connected to $x$ by some path.
  The homotopy $\Theta$ is thus given by
  \[
    \begin{cases}
      \Theta(0, x) = x \\
      \Theta(1, x) = g([x])
    \end{cases}
  \]
  and is tracked by $\lambda \langle i, n \rangle .
  \alpha(n) \cdot 1$.
\end{proof}

\refprop{discr_reflection_htpy_eq} actually gives us a concrete description of the homotopy category of fibrant assemblies.
Since every assembly is homotopy-equivalent to a modest set (the discrete reflection), fibrant assemblies and fibrant modest sets are identified in $Ho(\Asm_f)$.
This immediately gives us:

\begin{prop}{asm_ho_mod}
  The homotopy category of fibrant assemblies $Ho(\Asm_f)$ is equivalent to the category of modest sets.
\end{prop}
\begin{proof}
  By \refprop{discr_reflection_htpy_eq}, every assembly $X$ is homotopy-equivalent to $X_d$.
  Furthermore, every modest set is fibrant by \refcoro{discrete_impl_fibrant}, and $X_d \in \Asm_f$.
  It is thus the case that $Ho(\Asm_f) \simeq Ho(\Mod)$.
  By \refprop{discrete_obj_no_nontriv_paths}, the category $\Mod$ has no non-trivial homotopies, therefore $Ho(\Mod) \simeq \Mod$.
  As a result, the homotopy category of fibrant assemblies is the category of modest sets.
\end{proof}

\subsection{Assemblies and the path object construction}
As an application of \refprop{char_fibrant_asm}, we would like to present a comparison with the path object construction of Van Oosten \cite{htpyEff}.
Van Oosten presented a path object category \cite{pathObjCat} structure on the effective topos.
In his setting, the object of paths in $(X, \sim)$ is represented not by an exponent $(X, \sim)^\mathbb{I}$, but by a different object $\mathsf{P}(X, \sim)$, which is built out of paths of ``various length'': $I_n$ defined below. Whilst such an object is generally different from $(X, \sim)^{\mathbb{I}}$, we can show that both constructions are equivalent if $X$ is a fibrant assembly.
We refer the reader to the original paper for the detailed definitions.

\begin{definition}
  \label{def:interval_n}
  An assembly $I_n$ is defined to be an underlying set $\{0, \dots, n\}$ with the realizability relation
  $E(i) = \{ i, i +1 \}$. Note that $I_1$ is isomorphic to $\mathbb{I}$.
\end{definition}

\begin{definition}[{\cite[Definition 1.3]{htpyEff}}]
  A map $\sigma : I_n \to I_m$ is order and endpoint preserving iff
  \begin{enumerate}
  \item $\sigma(i) \leq \sigma(j)$ whenever $i \leq j$
  \item $\sigma(0) = 0$ and $\sigma(n) = m$
  \end{enumerate}
\end{definition}

\begin{definition}[{\cite[Definition 1.5]{htpyEff}}]
  Given an assembly $(X, E)$ the path object $\mathsf{P}(X, E)$
  (denoted as $\mathsf{P}(X)$ when unambiguous) is an assembly

  \begin{enumerate}
  \item With the underlying set being a quotient of $\{ (n, f) \mid f : I_n \to X \}$ by the relation $\sim$, defined as:
    $(n, f) \sim (m, g)$ if one of the following conditions hold
    \begin{enumerate}
    \item $n \geq m$ and there is an order and endpoint preserving map
      $\sigma : I_n \to I_m$ such that $f = g \sigma$; or
    \item $m \geq n$ and there is an order and endpoint preserving map
      $\sigma : I_m \to I_n$ such that $g = f \sigma$.
    \end{enumerate}
  \item With the realizability relation given by
    $E_{\mathsf{P}(X)}([(n, f)]) = \bigcup_{(m, g) \in [(n, f)]} \{ \langle m, b \rangle \mid b \Vdash_{X^{I_m}} g \}$
  \end{enumerate}
\end{definition}

\begin{proposition}
  Suppose $X$ is a fibrant assembly. Then $\mathsf{P}(X)$ is homotopic to $X^{\mathbb{I}}$
\end{proposition}
\begin{proof}
Given an $n$-path $[(n, q)] \in \mathsf{P}(X)$ one can, by repeated application of the composition, obtain a path $\mathfrak{p}(q) : \mathbb{I} \to X$ such that $\mathfrak{p}(q)(0) = q(0)$ and $\mathfrak{p}(q)(1) = q(n)$.
Furthermore, by \refprop{char_fibrant_asm}, there is a recursive $\alpha$ witnessing the uniformity of $s : X^{\mathbb{I}} \to X$.
That means that given a realizer $m \in E_X(q(0))$, the term $\lambda x.
\alpha(m)$ tracks $\mathfrak{p}(q)$.
One can obtain such $m$ using the realizer for the original $q$.

This defines a map $\mathfrak{p} : \mathsf{P}(X) \to X^{\mathbb{I}}$, for a fibrant assembly $X$.
One can check that a map $i : X^{\mathbb{I}} \to \mathsf{P}(X)$ that embeds $X^{\mathbb{I}}$ into the path object $\mathsf{P}(X)$ by sending $p : {\mathbb{I}} \to X$ to $[(1, p)] \in \mathsf{P}(X)$ is a right inverse of $\mathfrak{p}$.
We can show that it is also a left homotopy inverse of $\mathfrak{p}$.

We do so by defining a homotopy $\theta : {\mathbb{I}} \times \mathsf{P}(X) \to X$ as $\theta(0, [(n, q)]) = [(n, q)]$ and $\theta(1, [(n, q)]) = [(1, \mathfrak{p}(q))]$.
What remains is to provide a common realizer for $[(n, q)]$ and $[(1, \mathfrak{p}(q))]$ uniformly, given a realizer for $[(n, q)]$.
From a realizer of $[(n, q)]$ one can find a realizer $k \in E_X(q(0)) = E_X(\mathfrak{p}(q)(0))$.
Using the fibrancy of $X$ one can find a realizer $\alpha(k) \in \bigcap_{x' \in [x]} E_X(x')$.
Then $\lambda x.
\alpha(k)$ realizes both $q : I_n \to X$ and $\mathfrak{p}(q) : {\mathbb{I}} \to X$.
\end{proof}

\section{Conclusions and future research directions}
\label{sec:conclusion}


\subsection{Summary}

We have presented a way of obtaining a model structure on a full subcategory of a general topos, starting from a an interval object $\mathbb{I}$ dominance $\Sigma$ which contains the endpoint inclusion map $2 \to \mathbb{I}$.
The resulting model structure is sufficient for interpreting Martin-L\"of type theory with intensional identity types--which are interpreted with the help of the interval object.
The resulting model of type theory supports $\Pi$- and $\Sigma$-types, and functional extensionality holds for $\Pi$-types.

We have worked out the construction in the case of the effective topos $\Eff$.
For this model structure we have obtained some results characterizing contractible objects and maps, as well as fibrant assemblies.

\subsection{Future research questions} There remains several directions which can be further explored.
One of the most interesting questions would be extending the model category structure on $\mathcal{E}_f$ to the whole topos $\mathcal{E}$.
The mapping cocylinder construction in Section \ref{sec:wfs_triv_cof_fib} would not carry over directly, so one would have to find another way of constructing an (acyclic cofibrations, fibrations) weak factorisation system.

In this work we have decided to politely side-step the issues of coherence (as discussed in, e.g., 
\cite{Curien:LCCC}).
The authors expect that it is possible to resolve the coherence issues by considering algebraic counterparts of the homotopy-theoretic notions considered in this paper, such as algebraic weak factorisation systems (as done in the work of Gambino and Sattler \cite{UniformFibrations}) and algebraic model structures \cite{riehl2011algebraic}, but this issue should be investigated further.

In addition, there are several open questions regarding the concrete model $\Eff_f$ presented in this paper.
As already mentioned, it is unknown to the authors whether there is an object in $\Eff_f$ that has non-trivial higher homotopies.
It is clear, however, that such an object has to live outside the category of assemblies, as all assemblies are h-sets.
In general, is there a nice way of constructing higher inductive types in the model?
And if so, could the discrete reflection play the role of 0-truncation?
Extending the model to the whole of $\Eff$ might solve this problem, as we would be able to
consider fibrant replacements of objects with non-trivial homotopies.

Another interesting aspect of the effective topos is the existence of an internal small complete category of modest sets \cite{Hyland:smallCategory}, which is represented by a \emph{universal family of modest sets}.
Such a internal category which can be used as a type universe for interpreting second-order $\lambda$-calculus \cite{StreicherT:semttc}.
Unfortunately, by Example \ref{example:discrete_map_no_fib} this universal family cannot be a fibration.
The natural question to ask is then the following: does there exist a map $u$ which is a fibration, discrete, and has a fibrant codomain, such that every discrete map that is a fibration is a pullback of $u$?
And if so, is this universal fibration univalent?

Finally, it remains to be seen how much of the theory carries over to other realizability toposes.


\bibliographystyle{plain} \bibliography{modstr}

\end{document}